\documentclass[12pt]{article}

\renewcommand{\title}[1]{
\begin{center} \Large \bf #1 \end{center}
}

\renewcommand{\author}[2]{
 \begin{center} #1  \vspace{3mm} \\
  #2 \\
 \end{center}
\addvspace{\baselineskip}
}

\usepackage{amssymb}
\usepackage{amsmath}

\usepackage{amsthm}
\newtheorem{thm}{Theorem}[section]
\newtheorem{prop}[thm]{Proposition}
\newtheorem{cor}[thm]{Corollary}
\newtheorem{lem}[thm]{Lemma}

\theoremstyle{definition}

\theoremstyle{remark}

\makeatletter
\@addtoreset{equation}{section}

\makeatother

\begin{document}

\baselineskip 5mm

\title{
Gauge Theories in Noncommutative
Homogeneous K\"ahler Manifolds
%
}

\author{${}^{1,2}$Yoshiaki Maeda, ${}^3$ Akifumi Sako, ${}^4$ 
Toshiya Suzuki and~ ${}^4$
Hiroshi Umetsu }{
${}^{1}$ Department of Mathematics,
Faculty of Science and
 Technology, Keio University\\
3-14-1 Hiyoshi, Kohoku-ku, Yokohama 223-8522, Japan\\
${}^2$
Mathematical Research Centre, University of Warwick\\
Coventry CV4 7AL, United Kingdom,\\ 
${}^3$ 
Department of Mathematics,
Faculty of Science Division II,\\
Tokyo University of Science,
1-3 Kagurazaka, Shinjuku-ku, Tokyo 162-8601, Japan\\ 
${}^4$
Kushiro National College of Technology\\
2-32-1 Otanoshike-nishi, Kushiro, Hokkaido 084-0916, Japan }

\vspace{1cm}

\abstract{ 
We construct a gauge theory on 
a noncommutative homogeneous K\"ahler manifold,
where we employ the deformation quantization 
with separation of variables for K\"ahler manifolds 
formulated by Karabegov.
A key point in this construction 
is to obtaining vector fields which act as inner derivations
for the deformation quantization.
We give an explicit construction of  this gauge 
theory on noncommutative
${\mathbb C}P^N$ and noncommutative ${\mathbb C}H^N$.}


\section{Introduction}

Field theories on noncommutative spaces appear in various phenomena 
in physics.
For example, 
effective theories on D-branes with NS-NS B field backgrounds give rise to 
gauge theories on noncommutative spaces \cite{Seiberg:1999vs}. 
As another example, in matrix models \cite{Banks:1996vh,Ishibashi:1996xs}, 
noncommutative field theories corresponding to fuzzy spaces  appear when one expands the models around some classical solutions.

A typical noncommutative space is the noncommutative ${\mathbb R}^d$.
Field theories on the noncommutative ${\mathbb R}^d$
have many intriguing properties. For example, there is work on 
the existence of noncommutative
instantons \cite{Nekrasov:1998ss},
noncommutative scalar solitons
\cite{Gopakumar:2000zd}, etc. as classical solutions 
and the appearance of UV-IR mixing \cite{Minwalla:1999px}
at the quantum level
(see also the review papers  \cite{Douglas:2001ba,sako_review,Szabo:2001kg}, for examples).
It is important to investigate whether
field theories on more generic noncommutative manifolds have similar
properties.  However, field theories on noncommutative manifolds are not
well understood at present, except for a few examples such as, noncommutative tori,
$S^2$, etc.

Several methods to construct noncommutative manifolds have been proposed,
including the important approach by the deformation
quantization.
Deformation quantization was first introduced in 
\cite{Bayen:1977ha}. 
After \cite{Bayen:1977ha}, several alternative methods of deformation
quantization were proposed \cite{DeW-Lec, Omori, Fedosov, Kontsevich}.
In particular, deformation quantization of K\"ahler manifolds was
studied in \cite{Moreno86a, Moreno86b, Cahen93, Cahen95}.
We study gauge field theories on noncommutative
K\"ahler manifolds based on the deformation
quantization with separation of variables  introduced by
Karabegov to quantize K\"ahler manifolds
\cite{Karabegov,Karabegov1996,Karabegov2011}.

The purpose of this paper is to construct gauge
theories on noncommutative homogeneous K\"ahler manifolds.
Field theories 
need to define differentials on base spaces.
Note that the usual differentiations by coordinates in
a noncommutative space may
not be derivations; in other words, they do not
satisfy the Leibniz rule for star products in general.
We use  inner derivations as differentials,
which are defined by commutators with a function $P$ under a star product,
{\it i.e.} $[ P , \ \cdot \ ]_*$. These operators automatically satisfy the Leibniz rule. 
For a generic $P$, the inner derivation $[ P , \ \cdot \ ]_*$ includes 
higher derivative terms.
The necessary and sufficient condition on $P$ such that the inner derivation includes no
higher derivative terms is known \cite{Muller:2004}.
The necessary and sufficient condition 
 is that $P$ is the Killing potential.

For homogeneous K\"ahler manifolds ${\cal G} / {\cal H}$,
there are Killing vectors ${\cal L}_a$ which constitute the Lie
algebra of the isometry group ${\cal G}$.
The Killing potential $P_a$  
corresponding to ${\cal L}_a$ exists, and ${\cal L}_a$ is represented by
the inner derivation 
${\cal L}_a = \{ P_a , \ \cdot \ \} 
= - \frac{i}{\hbar} [ P_a, \ \cdot \ ]_*$.

Using these Killing potentials, we construct a
gauge theory on  noncommutative homogeneous
K\"ahler manifolds.
In our previous papers \cite{Sako:2012ws,Sako:2013noa},
we studied deformation quantizations with
separation of variables for ${\mathbb C}P^N$ and ${\mathbb C}H^N$,
and gave explicit expressions for the star products.
Using these results, 
we describe $U(n)$ gauge theories on noncommutative 
${\mathbb C}P^N$ and on
noncommutative ${\mathbb C}H^N$, 
as examples.~ 
(On other types of noncommutative ${\mathbb C}P^N$, 
different gauge theories have been constructed.
For example, a gauge theory on fuzzy ${\mathbb C}P^N$ is studied 
in \cite{CarowWatamura:1998jn,Grosse:2004wm}. )

The organization of this article is as follows.  In section 2, after we
review deformation quantization with separation of variables
for K\"ahler manifolds proposed by Karabegov, we study differentials
on noncommutative K\"ahler manifolds. The conditions under which inner
derivations become vector fields (Killing vector fields) are
provided. We then construct gauge theories on 
noncommutative homogeneous
K\"ahler manifolds. 
In section 3, we discuss gauge theories on noncommutative 
${\mathbb C}P^N$ and ${\mathbb C}H^N$, as concrete examples. 
In section 4, we summarize our results and give some further discussion.


\section{Deformation quantization of gauge theories with separation of
 variables}

\subsection{Deformation quantization with separation of variables}
\label{reviewKarabegov}
We briefly review the deformation quantization with separation of variables
for K\"ahler manifolds, which proposed by Karabegov \cite{Karabegov1996}.

Let $\Phi$ be a K{\" a}hler potential and $\omega$
a K{\" a}hler 2-form for $N$-dimensional K\"ahler manifolds $M$:
\begin{equation}
 \omega := i g_{k \bar{l}} dz^{k} \wedge d \bar{z}^{l}, ~~~~
  g_{k \bar{l}} := 
  \frac{\partial^2 \Phi}{\partial z^{k} \partial \bar{z}^{l}} .
\end{equation}
We denote the inverse of the metric $(g_{k \bar{l}})$ as 
$(g^{\bar{k} l})$, 
and set $g_{\bar{k}l} = g_{l \bar{k}}$,  
$g^{l \bar{k}} =  g^{\bar{k} l} $.
We use the following abbreviations
\begin{align}
\partial_k = \frac{\partial}{\partial z^{k}} , ~~~~
\partial_{\bar{k}} = \frac{\partial}{\partial \bar{z}^{k}}.
\end{align} 
 
Deformation quantization is defined as follows.
Let $\cal F$ be a set of formal power series in $\hbar$ with
coefficients of $C^{\infty}$
functions on $M$
\begin{eqnarray}
{\cal F} := \left\{  f \ \Big| \ 
f = \sum_k \hbar^k f_k, ~f_k \in C^\infty (M)
\right\} ,
\end{eqnarray}
where $\hbar$ is a noncommutative parameter.
A star product is defined on ${\cal F}$ by 
\begin{eqnarray}
f * g = \sum_k \hbar^k C_k (f,g), 
\end{eqnarray}
such that the product satisfies the following conditions.
\begin{enumerate}
\item $*$ is associative product.
\item $C_k$ is a bidifferential operator.
\item $C_0$ and $C_1$ are defined as 
\begin{eqnarray}
&& C_0 (f,g) = f g,  \\
&&C_1(f,g)-C_1(g,f) = i \{ f, g \}, \label{weakdeformation}
\end{eqnarray}
where $\{ f, g \}$ is the Poisson bracket.
\item $ f * 1 = 1* f = f$.
\end{enumerate}

Moreover, $*$ is called a star product with separation of variables when it satisfies
\begin{equation}
 a * f = a f, ~~~~ f * b = f b, 
\end{equation} 
for any holomorphic function $a$ and any anti-holomorphic function $b$.
Karabegov constructed 
a star product with separation of variables for K\"ahler manifolds
in terms of differential operators  \cite{Karabegov,Karabegov1996},
as briefly explained  below.
For the left star multiplication by $f \in {\cal F}$, 
there exists  a differential
operator $L_f$ such that
\begin{equation}
 L_f g = f * g .
\end{equation}
$L_f$ is given as a formal power series in $\hbar$
\begin{equation}
 L_f = \sum_{n=0}^{\infty} \hbar^n A^{(n)},
  \label{Lf-An}
\end{equation}
where $A^{(n)}$ is a differential operator which contains only partial
derivatives  by $z^i$ and has the following form
\begin{equation}
 A^{(n)} = \sum_{k\geq 0} a^{(n;k)}_{\bar{i}_1 \cdots \bar{i}_k} 
  D^{\bar{i}_1} \cdots D^{\bar{i}_k},
  \label{An-a}
\end{equation}
where
\begin{equation}
 D^{\bar i} = g^{{\bar i} j} \partial_j.
\end{equation}
In particular, $a^{(n;0)}$ which is a $C^{\infty}$ function on $M$ acts as a multiplication operator.
Note that the differential operators $D^{\bar i}$ satisfy the following relations,
\begin{align}
 [ D^{\bar i} , D^{\bar j} ] &= 0, \\
 [D^{\bar{i}}, \partial_{\bar{j}}\Phi] &= \delta_{ij}.
 \label{D-dphi}
\end{align}

Karabegov showed the following theorem.
\begin{thm}[\cite{Karabegov,Karabegov1996}]
$L_f$ is uniquely determined by requiring the following conditions,
\begin{align}
 L_f 1 = f * 1 =f, 
 \label{Lf-dphi1}\\
 [L_f , \partial_{\bar i} \Phi + \hbar \partial_{\bar i}] = 0.
 \label{Lf-dphi2}
\end{align}
\end{thm}
Substituting the expression of $L_f$ in (\ref{Lf-An}) to  the conditions
(\ref{Lf-dphi1}) and (\ref{Lf-dphi2}), one obtains the following recursion
relations,
\begin{align}
 & A^{(0)} = f, ~~~ A^{(r)} 1 = 0, 
 \label{Ar-rec1} \\
 & [A^{(r)}, \partial_{\bar{i}}\Phi] = [\partial_{\bar{i}}, A^{(r-1)}],
 \label{Ar-rec2}
\end{align}
for $r \geq 1$, where $f$ is assumed to be independent of
$\hbar$ (in general, we set for $f \in {\cal F}$, $A^{(0)} = f_0$ and
$A^{(r)} 1 = f_r$ in the above equations). 
In the case of $r=1$, one can easily find
\begin{equation}
 A^{(1)} = \partial_{\bar{i}} f D^{\bar{i}},
  \label{A1}
\end{equation}
where (\ref{D-dphi}) is used. 
Let us observe that $a^{(r;0)} = a^{(r;1)}_{\bar{i}}= 0$ for $r\geq 2$ in
the expressions (\ref{An-a}), namely, 
\begin{equation}
 A^{(r)} = \sum_{k\geq 2} a^{(r;k)}_{\bar{i}_1 \cdots \bar{i}_k} 
  D^{\bar{i}_1} \cdots D^{\bar{i}_k}, ~~~~ r\geq 2.
  \label{Ar}
\end{equation}
 From the condition (\ref{Ar-rec1}), $a^{(r;0)}=0 ~(r\geq 1)$ trivially
obeys.  We then define the twisted symbol of $A^{(r)}$ as 
$a^{(r)} (\xi) = \sum a^{(r;k)}_{\bar{i}_1 \cdots \bar{i}_n} 
\xi^{\bar{i}_1} \cdots \xi^{\bar{i}_n}$. 
The twisted symbol of the left hand side in
(\ref{Ar-rec2}) is $\partial a^{(r)}(\xi)/\partial \xi^{\bar{i}}$ from
(\ref{D-dphi}).  That of the right hand side in (\ref{Ar-rec2}) does not
contain the zeroth order term of $\xi$, because of $a^{(r;0)}=0$ for
$r\geq 1$. Therefore,
$a^{(r)} ~(r\geq 2)$ does not contain the first order term of $\xi$. 
This prove the assertion.

Here is a useful theorem given by Karabegov.
\begin{thm}[\cite{Karabegov,Karabegov1996}]
The differential operator $L_f$ for an arbitrary function $f$ is
obtained from the operator $L_{{\bar z}^i}$, which corresponds to
the left $*$ multiplication of ${\bar z}^i$,
\begin{equation} \label{Lf_Lz}
 L_f = \sum_{\alpha} \frac{1}{\alpha !} 
  \left(\frac{\partial}{\partial {\bar z}}\right)^\alpha 
  f (L_{\bar z} - {\bar z})^\alpha , 
\end{equation} 
where $\alpha$ is a multi-index.
\end{thm}

Similarly, the differential operator 
$\displaystyle{R_f = \sum_{n=0}^\infty \hbar^n B^{(n)}}$ corresponding to
the right $*$ multiplication by a function $f$ contains only partial
derivatives by $\bar{z}^i$ and is determined by the conditions
\begin{align}
 & R_f 1 = 1*f = f, \\
 & [R_f, \partial_i \Phi + \hbar{\partial_i}] = 0.
\end{align}
$B^{(n)}$ has the following form,
\begin{align}
 B^{(0)} &= f, ~~~ B^{(1)} = \partial_i f D^i, ~~~
  B^{(r)} = \sum_{k\geq 2} b^{(r;k)}_{i_1 \cdots i_k} 
 D^{i_1} \cdots D^{i_k},
 \label{Br}
\end{align}
where $D^i = g^{i\bar{j}}\partial_{\bar{j}}$.
The differential operator $R_f$ for an arbitrary function $f$ is
obtained from the operator $R_{z^i}$, which corresponds to
the right $*$ multiplication by $z^i$,
\begin{equation} \label{Rf_Rz}
 R_f = \sum_{\alpha} \frac{1}{\alpha !} 
  \left(\frac{\partial}{\partial z}\right)^\alpha 
  f (R_z - z)^\alpha. 
\end{equation}

\subsection{Derivations in deformation quantization }\label{Deri_Defo}

A differential calculus on noncommutative spaces  
can be constructed based on the derivations 
of the algebra $C^\infty(M)[[\hbar]]$ with its star product, 
whose derivation ${\mathbf d}$ 
are linear operators satisfying the Leibniz rule, i.e.
$ {\mathbf d} (f*g)= {\mathbf d} f * g + f * {\mathbf d}g$ .
In commutative space, vector fields are obviously derivations. 
However  first order 
differential operators in noncommutative space do not 
satisfy the Leibniz rule in general.
In this 
subsection, we study inner derivations ${\cal L}$,
in particular, let ${\cal L}$ be a linear differential operator
such that 
${\cal L} (f) =[P, f]_* := P*f - f*P,$ ($P, f \in C^\infty(M)[[\hbar]]$).

Note that
inner derivations are not first order differential operator, 
since the
explicit expression of the star commutator $[P, f]_*$ includes higher
derivative terms of $f$ for a generic $P$.
In particular, inner derivations corresponding to vector fields play
an important role, when we construct field theories on noncommutative
spaces. 
It is known that such vector fields are given as the Killing vector fields \cite{Muller:2004}.
In this section, we review the fact to obtain the differentials on noncommutative
${\mathbb C}P^N$ and ${\mathbb C}H^N$.

\begin{prop}
Let $M$ be a K\"ahler manifold with the $*$ product 
with separation of variables given in section
 \ref{reviewKarabegov}. 
Let $P \in C^{\infty}(M)[[ \hbar ]]$, 
$f$ be an arbitrary $C^{\infty}$ function on $M$ and 
$[ P , f ]=P*f - f*P $ i.e. the inner derivation of the $*$-product
mentioned above. 
Then $[P, f]_* = i \hbar\{P, f\}$ if and only if $D^i D^j P =0$ and 
$D^{\bar{i}} D^{\bar{j}} P=0$ for all $i, j = 1, 2, \cdots, N$. 
Namely, higher derivative terms of $f$ in $[P, f]_*$ vanish and this inner
derivation is given by some vector field when these 
conditions are satisfied.
\end{prop}
\begin{proof}
From the formulas (\ref{Lf_Lz}) and (\ref{Rf_Rz}), we find
\begin{align}
 [P, f]_* &= R_f P - L_f P \nonumber \\
 &= \sum_\alpha \left[ \frac{1}{\alpha !}
 \left(R_z - z\right)^\alpha P \cdot 
 \left(\frac{\partial}{\partial z} \right)^\alpha
 -(L_{\bar{z}} - \bar{z})^\alpha P \cdot
 \left(\frac{\partial}{\partial \bar{z}} \right)^\alpha
 \right] f.
 \label{[P,f]}
\end{align}
The differential operators $L_{\bar{z}^i}$ and $R_{z^i}$ have the
following forms,
\begin{align}
  L_{\bar{z}^i} &= \bar{z}^i 
 + \sum_{n=1}^\infty \hbar^n A^{(n)}_{\bar{i}}, \\
 R_{z^i} &= z^i + \sum_{n=1}^\infty \hbar^n B^{(n)}_i.
\end{align}
From (\ref{A1}), (\ref{Ar}) and (\ref{Br}), $A^{(n)}_{\bar{i}}$ 
and $B^{(n)}_i$ 
are given as
\begin{align}
 A^{(1)}_{\bar{i}} &= D^{\bar{i}}, ~~~~
 A^{(r)}_{\bar{i}} = \sum_{k\geq 2} 
 a^{(r;k)}_{\bar{i}; \bar{j}_1 \cdots \bar{j}_k} 
  D^{\bar{j}_1} \cdots D^{\bar{j}_k}, ~~~~ r\geq 2, 
 \label{Air} \\
 B^{(1)}_i &= D^i, ~~~~
 B^{(r)}_i = \sum_{k\geq 2} b^{(r;k)}_{i; j_1 \cdots j_k} 
  D^{j_1} \cdots D^{j_k}, ~~~~ r\geq 2.
 \label{Bir}
\end{align}
The first order terms in $\hbar$ in the right hand side of
(\ref{[P,f]}) give the Poisson bracket $i\hbar\{P, f\}$.  
Looking at 
$\left(L_{\bar{z}^{i_1}} - \bar{z}^{i_1}\right) \cdots
\left(L_{\bar{z}^{i_k}} - \bar{z}^{i_k}\right) P$ for $k \geq 2$ 
and $P=\sum_{n=0}^\infty \hbar^n P^{(n)}$, we have
\begin{align}
\left(L_{\bar{z}^{i_1}} - \bar{z}^{i_1}\right)
\cdots \left(L_{\bar{z}^{i_k}} - \bar{z}^{i_k}\right) P
 &= \sum_{n=0}^\infty \sum_{m_1=1}^\infty \cdots
 \sum_{m_k=1}^\infty \hbar^{n+m_1+\cdots m_k} 
 A^{(m_1)}_{\bar{i}_1} \cdots A^{(m_k)}_{\bar{i}_k} P^{(n)}.
 \label{LLP} 
\end{align}

Assuming $[P, f]=i\hbar\{P, f\}$, namely, assuming  that the all terms in
 (\ref{LLP}) vanish, we show $D^{\bar{i}} D^{\bar{j}} P=0$.
The terms of the order $\hbar^2$ in (\ref{LLP}) exists only for
$k=2$ and has the following form,
\begin{align}
 A_{\bar{i}_1}^{(1)}A_{\bar{i}_2}^{(1)} P^{(0)} 
 = D^{\bar{i}_1} D^{\bar{i}_2}P^{(0)}.
\end{align}
Hence, $D^{\bar{i}_1}D^{\bar{i}_2}P^{(0)}=0$, and we find
\begin{align}
 \sum_{m_1=1}^\infty \cdots
 \sum_{m_k=1}^\infty \hbar^{n+m_1+\cdots m_k} 
 A^{(m_1)}_{\bar{i}_1} \cdots A^{(m_k)}_{\bar{i}_k} P^{(0)} = 0,
\end{align}
from the explicit forms of $A_{\bar{i}}^{(r)}$, (\ref{Air}).
As the induction assumption, we set 
$D^{\bar{i}}D^{\bar{j}}P^{(n)}=0$ for $n= 0, 1, \dots, r-1$.
Similar to the case of $P^{(0)}$, the following equation holds for 
$n=0, 1, \dots, r-1$, 
\begin{align}
 \sum_{m_1=1}^\infty \cdots \sum_{m_k=1}^\infty \hbar^{n+m_1+\cdots m_k}
 A^{(m_1)}_{\bar{i}_1} \cdots A^{(m_k)}_{\bar{i}_k} P^{(n)} = 0,
\end{align}
and the right hand side of (\ref{LLP}) becomes
\begin{align}
 \sum_{n=r}^\infty \sum_{m_1=1}^\infty \cdots
 \sum_{m_k=1}^\infty \hbar^{n+m_1+\cdots m_k} 
 A^{(m_1)}_{\bar{i}_1} \cdots A^{(m_k)}_{\bar{i}_k} P^{(n)}.
\end{align} 
The term of the order ${\cal O}(\hbar^{r+2})$ in this sum exists
only for $k=2$ and has the following form,
\begin{align}
  A_{\bar{i}_1}^{(1)}A_{\bar{i}_2}^{(1)} P^{(r)} 
 = D^{\bar{i}_1} D^{\bar{i}_2}P^{(r)}.
\end{align}
Thus, $D^{\bar{i}_1}D^{\bar{i}_2}P^{(r)}=0$. 
Therefore, it is shown that $D^{\bar{i}}D^{\bar{j}}P=0$ holds for all $i, j$.

Similarly, $D^i D^j P =0$ can be derived by considering 
$\left(R_z - z\right)^\alpha P$.   

The converse is easily shown from the above equations.
\end{proof}

Real valued functions which satisfy $D^i D^j P =0$ and
$D^{\bar{i}}D^{\bar{j}} P=0$ on K\"ahler manifolds are known as Killing
potentials \cite{Freedman:2012zz}. The Killing potential gives a
holomorphic Killing vector 
$\zeta^i \partial_i + \zeta^{\bar{i}}\partial_{\bar{i}} 
= \{P, \cdot~\}$,
\begin{align}
 \zeta^i &= -ig^{i\bar{j}}\partial_{\bar{j}}P = -iD^i P, \\
 \zeta^{\bar{i}} &= ig^{\bar{i}j} \partial_j P = iD^{\bar{i}}P .
\end{align} 
$\zeta^i$ is holomorphic, and $\zeta^{\bar{i}}$ is anti-holomorphic.
The metric
and the complex structure of the K\"ahler manifold are invariant under
the transformations generated by the holomorphic Killing vectors, 
$\delta_\zeta z^i = \zeta^i, ~\delta_\zeta \bar{z}^i = \zeta^{\bar{i}}$.
Summarizing these facts, we have the following corollary
\begin{cor}{\rm (\cite{Muller:2004})}
In deformation quantization defined in Section \ref{reviewKarabegov}, 
the inner derivations given as vector fields are 
the Killing vector fields 
${\cal L}_a = \zeta_a^i\partial_i + \zeta_a^{\bar{i}} \partial_{\bar{i}}$ .
\end{cor}


\subsection{Deformed gauge theory}
In the previous section, we studied inner derivations given as vector fields on
noncommutative K\"ahler manifolds.  Using this, we
investigate gauge theories with a gauge
group $G$ on noncommutative 
homogeneous K\"ahler manifolds $M$
given by the deformation quantization
in Section \ref{reviewKarabegov}. ~
( From several view points, matrix models and its related topics studied in
\cite{Kawai:2009vb,Kawai:2010sf,Kitazawa} are useful 
for understanding the gauge 
theory constructed in this subsection.)~
In the following, we
consider $U(n)$ gauge theories for simplicity.  
All results in
this section can be applied for any matrix groups.

At first, we introduce a noncommutative $U(n)$ transformations
as a deformation of the unitary transformations.
If $g \in U(n)$, then 
$g^\dagger g = I$,
where $g^{\dagger}$ is the hermitian conjugate of $g$ 
and $I$ is the identity matrix.
As a natural extension, we define
$G:= C^{\infty} (M)[[ \hbar ]]  \otimes GL(n ; {\mathbb C})$
such that 
for $\displaystyle U= \sum_{k=0}^{\infty} \hbar^k U^{(k)}$ and
$\displaystyle U^\dagger = \sum_{k=0}^{\infty} \hbar^k U^{(k)\dagger}
\in G$, 
\begin{align} \label{3_1}
 U^{\dagger} * U = \sum_{n=0}^\infty \hbar^n 
 \sum_{m=0}^n U^{(m)\dagger} * U^{(n-m)} = I. 
\end{align}
This condition is imposed for each order of $\hbar$.
For arbitrary $U^{(0)} \in C^\infty (M)\otimes  U(n) $, 
(\ref{3_1}) has solutions which are determined
recursively at each order of $\hbar$ \cite{maeda_sako}.

In noncommutative K\"ahler manifolds, the partial derivative
$\partial$ does not play an essential role, since the
Leibniz rule is failed;   
$\partial (f * g) \neq \partial f * g + f * \partial g$.  
To construct a
covariant derivative of a gauge theory, we should adopt some derivations
(operators which satisfy the Leibniz rule) instead of $\partial$.  In
particular, inner derivations are given by commutators of the star
product. The space of inner derivations is infinite dimensional. Hence,
if the whole space of inner derivations is used to construct gauge
theories, the infinite number of gauge fields would be introduced. (See
for example \cite{DuboisViolette:1988ps,DuboisViolette:1999cj} .)  
In this article, we
consider deformation quantization of a homogeneous K\"ahler manifold
${\cal G}/{\cal H}$ and choose a subalgebra of the Lie algebra of inner
derivations. 
Here, we assume that ${\cal G}$ is a connected semisimple Lie
group so that ${\cal G}/{\cal H}$ has at least nondegenerate metric. 
Then, we construct a deformation quantization of gauge
theories on ${\cal G}/{\cal H}$ whose covariant derivatives are derived
from inner derivations corresponding to the Killing vector fields.

In a homogeneous K\"ahler manifold ${\cal G}/{\cal H}$, there are the
holomorphic Killing vector fields 
${\cal L}_a = \zeta_a^i(z)\partial_i + 
\zeta_a^{\bar{i}}(\bar{z}) \partial_{\bar{i}}$ 
corresponding to the Lie
algebra of the isometry group ${\cal G}$,
\begin{equation}
 [{\cal L}_a, {\cal L}_b] = if_{abc} {\cal L}_c,
\end{equation}
where $a$ is an index of the Lie algebra of ${\cal G}$ and $f_{abc}$ is
its structure constant.
There exists the Killing potential $P_a$ corresponding to ${\cal L}_a$, 
$ {\cal L}_a = \{P_a, \cdot~\}$.
As stated in the previous section, the Killing vector ${\cal L}_a$ can
be described by $*$-commutator and satisfy the Leibniz rule, 
\begin{align}
 {\cal L}_a &= -\frac{i}{\hbar} [P_a, \cdot~]_*, \\
 {\cal L}_a (f*g) &= ({\cal L}_a f) * g + f * ({\cal L}_a g).
\end{align}
The Killing vectors are normalized here as
\begin{equation}
\eta^{ab} \zeta_a^i \zeta_b^{\bar{j}} = g^{i\bar{j}}, ~~~~
\eta^{ab} \zeta_a^i \zeta_b^j = 0, ~~~~
\eta^{ab} \zeta_a^{\bar{i}} \zeta_b^{\bar{j}} = 0, 
\label{zeta-g}
\end{equation} 
where $\eta^{ab}$ is the inverse of the Killing form of the Lie algebra
of ${\cal G}$. 
We introduce gauge fields corresponding to ${\cal L}_a$ in the following.

Let us consider a commutative homogeneous K\"ahler manifold 
$M={\cal G}/{\cal H}$.  
We denote the indices of $TM$ as $\mu = 1, 2, \cdots, 2N$
for combining the holomorphic and anti-holomorphic indices.  We define
${\cal A}_a^{(0)}$ as
\begin{align}
 {\cal A}_a^{(0)} = \zeta^\mu_a  A_\mu 
 = \zeta^i_a A_i + \zeta^{\bar{i}}_a A_{\bar{i}},
\end{align}
where $A_i$ and $A_{\bar{i}}$ are gauge fields on $M$.
Its curvature is defined as
\begin{align}
 {\cal F}_{ab}^{(0)}:=
 {\cal L}_a {\cal A}_b^{(0)}-
 {\cal L}_b {\cal A}_a^{(0)}
 -i [ {\cal A}_a^{(0)} , {\cal A}_b^{(0)} ] -i f_{abc}{\cal A}_c^{(0)}  ,
\end{align}
where $[ A , B ] = AB-BA$.
${\cal F}_{ab}^{(0)}$ is related to the curvature of $A_\mu$, 
$F_{\mu\nu} = \partial_\mu A_\nu - \partial_\mu A_\nu -i [A_\mu, A_\nu]$,
as 
\begin{align}
{\cal F}_{ab}^{(0)} = \zeta_a^\mu \zeta_b^\nu F_{\mu\nu}.
\end{align}
By using (\ref{zeta-g}), it is shown that 
\begin{align}\label{taioukankei}
\eta^{ac} \eta^{bd} {\cal F}_{ab}^{(0)} {\cal F}_{cd}^{(0)}
= g^{\mu\rho} g^{\nu\sigma} F_{\mu\nu} F_{\rho\sigma}.
\end{align}

Now, we consider a noncommutative deformation of gauge theories.
We define
\begin{align}
{\cal A}_a := \sum_{k=0}^\infty \hbar^k {\cal A}_a^{(k)}
\end{align}
as a gauge field, and define its gauge transformation by
\begin{align}
{\cal A}_a \rightarrow {\cal A}_a' = i U^{-1}* {\cal L}_a U
+U^{-1}* {\cal A}_a * U .   \label{Atrans}
\end{align}
Let us define a curvature of ${\cal A}_a$ by
\begin{align}
{\cal F}_{ab} :=
{\cal L}_a {\cal A}_b -
{\cal L}_b {\cal A}_a 
-i [ {\cal A}_a , {\cal A}_b ]_* -i f_{abc}{\cal A}_c . 
\end{align}
\begin{lem}
${\cal F}_{ab}$ transforms covariantly:
\begin{align}
{\cal F}_{ab} \rightarrow {\cal F}_{ab}' = U^{-1} * {\cal F}_{ab} *U. 
\label{Ftrans}
\end{align}
\end{lem}

\begin{proof}
 \begin{align}
  {\cal F}_{ab}' &=
  {\cal L}_a {\cal A}_b' -
  {\cal L}_b {\cal A}_a '
  -i [ {\cal A}_a' , {\cal A}_b' ]_* -i f_{abc}{\cal A}_c'.
  \label{F'}
 \end{align}
 Using (\ref{Atrans}) and 
 \begin{align*}
  0= {\cal L}_a (U^{-1} * U )=  {\cal L}_a U^{-1} * U+ U^{-1} * {\cal L}_a U 
 \end{align*}
 which is obtained from the Leibniz rule for ${\cal L}_a$,
 the right hand side of (\ref{F'}) is written as
 \begin{align*}
 U^{-1}*{\cal F}_{ab}*U + 
  i U^{-1} * [{\cal L}_a, {\cal L}_b] U
  + f_{abc} U^{-1} * {\cal L}_c U.
 \end{align*}
Noting that $[ {\cal L}_a , {\cal L}_b ]= i f_{abc} {\cal L}_c$,
we have ${\cal F}_{ab}' = U^{-1} * {\cal F}_{ab} *U$.
\end{proof}

Using this lemma, we obtain the gauge invariant action.
\begin{thm}
A gauge invariant action for the gauge field is given by
\begin{align}
 S_g :=
 \int_{{\cal G}/{\cal H}} \mu_g ~ 
 {\rm tr} \left(  \eta^{ac}\eta^{bd} 
 {\cal F}_{ab} * {\cal F}_{cd} \right),
\end{align}
where $\mu_g$ is a trace density.
\end{thm}
\begin{proof}
The gauge invariance of the action is obtained by
(\ref{Ftrans}) and the cyclic symmetry of the trace density.
The existence of trace density, $\int_M f*g \mu_g = \int_M g*f \mu_g$,
is guaranteed in \cite{Karabegov:1998hm}.
\end{proof}

Scalar fields are also introduced as similar to 
commutative case.
As an example, let us consider a complex scalar field
$\displaystyle \phi= \sum_k \phi^{(k)} \hbar^k$
and its hermitian conjugate $\displaystyle \phi^{\dagger}$
which transform as the fundamental representation of the gauge group,
\begin{align} \label{gauge_trans_scalar}
\phi \rightarrow \phi'= U^{-1} * \phi,~~~~
\phi^\dagger \rightarrow {\phi^{\dagger}}'=  \phi^{\dagger} * U .
\end{align}
A covariant derivative for this scalar field is defined by
\begin{align}
\nabla_a \phi := {\cal L}_a \phi - i {\cal A}_a * \phi, \label{covariavt scalar}
\end{align}
and then this transforms covariantly;
\begin{align}
{\nabla_a}' \phi' = U^{-1}* {\nabla_a} \phi .
\end{align}
Therefore we obtain 
the gauge invariant action.
\begin{thm}
 Let $\phi$ be a fundamental representation  complex scalar field 
and $\phi^{\dagger}$ be a hermitian conjugate of $\phi$
whose gauge transformations are given by
(\ref{gauge_trans_scalar}).
Then, the following action is gauge invariant.
\begin{align}
\label{scalar_action}
S_{\phi} = \int_{{\cal G}/{\cal H}} \mu_g
\left\{ \eta^{ab} \nabla_a \phi^{\dagger}* \nabla_b \phi
 + V(\phi^{\dagger} * \phi ) \right\} ,
\end{align}
where $V$ is a potential as a function of one variable.
\end{thm}



\section{
Gauge theories in noncommutative ${\mathbb C}P^N$ and ${\mathbb C}H^N$ }
In this section, as examples of the deformed gauge theories defined in
the previous section, we will construct noncommutative gauge theories on
${\mathbb C}P^N$ and ${\mathbb C}H^N$ by using deformation quantization
with separation variables.

 
\subsection{Deformation quantization with separation variables of
${\mathbb C}P^N$ and ${\mathbb C}H^N$}
\label{Defo_Qun_CPN_CHN}
We recall the results for the deformation quantization with
separation of variables for ${\mathbb C}P^N$ and ${\mathbb C}H^N$
\cite{Sako:2012ws}.

In the inhomogeneous coordinates $z^i ~(i=1, 2, \cdots, N)$, the K\"ahler
potential of $\mathbb{C}P^N$ is given by
\begin{align}
 \Phi = \ln \left(1+|z|^2\right), \label{phi}
\end{align}
where $|z|^2 = \sum_{k=1}^N z^k \bar{z}^k$.
The metric $(g_{i\bar{j}})$ is 
\begin{align}
 ds^2 &= 2g_{i\bar{j}}dz^id\bar{z}^j, \label{ds} \\
 g_{i\bar{j}} &= \partial_i \partial_{\bar{j}} \Phi
  = \frac{(1+|z|^2)\delta_{ij}-z^j \bar{z}^i}{(1+|z|^2)^2}, \label{metric}
\end{align}
and the inverse of the metric $(g^{\bar{i}j})$ is
\begin{align}
 g^{\bar{i}j} = (1+|z|^2)\left(\delta_{ij}+z^j\bar{z}^i\right). \label{inverse}
\end{align}

Recall that the left star multiplication for a function $f$, $L_f$, is
written by using $L_{{\bar z}^l}$, (\ref{Lf_Lz}). 
The explicit expression for $L_{{\bar z}^l}$ on $\mathbb{C}P^N$ is given by
\begin{align}
 L_{\bar{z}^l} &= \bar{z}^l + \hbar D^{\bar{l}}
  + \sum_{n=2}^\infty \hbar^n \sum_{m=2}^n a^{(n)}_m
  \partial_{\bar{j}_1}\Phi \cdots \partial_{\bar{j}_{m-1}}\Phi
  D^{\bar{j}_1} \cdots D^{\bar{j}_{m-1}} D^{\bar{l}}
  \nonumber \\
  &= \bar{z}^l
   + \sum_{m=1}^\infty \alpha_m(\hbar)
  \partial_{\bar{j}_1}\Phi \cdots \partial_{\bar{j}_{m-1}}\Phi
  D^{\bar{j}_1} \cdots D^{\bar{j}_{m-1}} D^{\bar{l}}, \label{L_z}
\end{align}
where
\begin{align}
 & \alpha_m (t) = \sum_{n=m}^\infty t^n a^{(n)}_m, \label{f-def} \\
 & \alpha_1 (t) = t, \\
 & \alpha_m(t) = t^m \prod_{n=1}^{m-1}\frac{1}{1-nt}
  = \frac{\Gamma(1-m+1/t)}{\Gamma(1+1/t)}, 
  \qquad (m\geq 2). 
 \label{fm}
\end{align}
The function $\alpha_m (t)$ actually coincides with the generating
function for the Stirling numbers of the second kind $S(n,k)$, and
$a^{(n)}_m$ is related to $S(n,k)$ as
\begin{align}
 a^{(n)}_{m} = S(n-1, m-1). \label{2ndS}
\end{align}
 
One of non-trivial star products is $\bar{z}^i*z^j$,
\begin{align}
 \bar{z}^i*z^j &= 
 \bar{z}^iz^j + \hbar \delta_{ij} (1+|z|^2)
  {}_2F_1\left(1, 1; 1-1/\hbar; -|z|^2\right) \nonumber \\
 & ~~~~+\frac{\hbar}{1-\hbar} \bar{z}^i z^j (1+|z|^2) 
 {}_2F_1 \left(1, 2; 2-1/\hbar; -|z|^2\right), \label{barz-z}
\end{align}
where ${}_2F_1$ is the Gauss hypergeometric function.

For $\mathbb{C}H^N$, similar results are obtained.
The K\"ahler potential and the metric are given by
\begin{align}
 \Phi =& - \ln \left(1-|z|^2\right), \\
 g_{i\bar{j}} &= \partial_i \partial_{\bar{j}} \Phi
  = \frac{(1-|z|^2)\delta_{ij}+\bar{z}^i z^j}{(1-|z|^2)^2}, \\
 g^{\bar{i}j} &= (1-|z|^2)\left(\delta_{ij}-\bar{z}^i z^j\right).
\end{align}

The operator $L_{\bar{z}^l}$ is
expanded as a power series of the noncommutative parameter $\hbar$,
and has  the following explicit representation,
\begin{align}
 L_{\bar{z}^l} &= \bar{z}^l + \hbar D^{\bar{l}}
  + \sum_{n=2}^\infty \hbar^n \sum_{m=2}^n 
 (-1)^{n-1} a^{(n)}_m
  \partial_{\bar{j}_1}\Phi \cdots \partial_{\bar{j}_{m-1}}\Phi
  D^{\bar{j}_1} \cdots D^{\bar{j}_{m-1}} D^{\bar{l}}
  \nonumber \\
  &= \bar{z}^l
   + \sum_{m=1}^\infty (-1)^{m-1} \beta_m(\hbar)
  \partial_{\bar{j}_1}\Phi \cdots \partial_{\bar{j}_{m-1}}\Phi
  D^{\bar{j}_1} \cdots D^{\bar{j}_{m-1}} D^{\bar{l}}, 
 \label{Lz-ch}
\end{align}
where 
\begin{align}
\beta_n(t)=(-1)^n \alpha_n(-t)= \frac{\Gamma(1/t)}{\Gamma(n+1/t)}.
\label{beta-n}
\end{align}
Then, one of non-trivial star products is $\bar{z}^i*z^j$,
\begin{align}
 \bar{z}^i*z^j =& 
  \bar{z}^iz^j + \hbar \delta_{ij} (1-|z|^2) 
 {}_2F_1\left(1, 1; 1 + 1/\hbar; |z|^2 \right)\nonumber \\
 & - \frac{\hbar}{1+\hbar} \bar{z}^i z^j (1-|z|^2)
 {}_2F_1 \left(1, 2; 2+1/\hbar; |z|^2\right).
\end{align}

We should comment on the relation between our previous results and those
of preceding related works
\cite{Balachandran,Bordemann,Hayasaka:2002db}.

Balachandran {\it et al.} gave an explicit expression of $*$ product on
fuzzy ${\mathbb C}P^n$, using  matrix regularization
\cite{Balachandran}.
Their $*$ product is expressed as a finite series.
Though our $*$ product is, in general, an infinite series in
$\hbar$, it coincides with Barachandran's $*$ product if we take
$\hbar = 1/L (L \in {\mathbb N})$.

On the other hand, 
Bordemann {\it et al.} obtained a star product which has a similar form of an
infinite series in the noncommutative parameter $\hbar$ to our star
product \cite{Bordemann}.
In fact, their star product is shown to be equivalent to ours 
(see \cite{Sako:2012ws} section 3).

Also in \cite{Hayasaka:2002db}, 
an explicit expression of a star product on fuzzy
$S^2$ is given as an infinite series in a noncommutative parameter,
which coincides with our expression in the case of ${\mathbb C}P^1$.

\subsection{Differentials on noncommutative $\mathbb{C}P^N$ } 
\label{KillingCP}
In this section, 
we study differentials in a noncommutative $\mathbb{C}P^N$ with
the star product with separation of variables.

In $\mathbb{C}P^N$, the conditions $D^i D^j P =0$ and 
$D^{\bar{i}}D^{\bar{j}}P=0$ can be solved as
\begin{align}
 P = \frac{\alpha_i z^i + \bar{\alpha}_i \bar{z}^i 
 + \beta_{ij}\bar{z}^i z^j}{1+|z|^2},
 \label{P-CP}
\end{align}
where $\alpha_i$ and $\beta_{ij} = \bar{\beta}_{ji}$ are complex
parameters and $|z|^2 = \sum_{i=1}^N z^i \bar{z}^i$.  The number of the real parameters is $N^2+2N$ and these
correspond to the $SU(N+1)$ isometry transformations of $\mathbb{C}P^N$.
In the following, we give concrete expressions of the Killing
potentials corresponding to the generators of $su(N+1)$, the Lie
algebra of $SU(N+1)$.

Homogeneous coordinates of $\mathbb{C}P^N$
\begin{align}
 \left\{\xi^A | A = 0, 1, \cdots, N \right\} &
 = \left\{ \xi^0, \xi^i | i= 1, 2, \dots, N \right\}
\end{align}
are related with inhomogeneous coordinates on the chart of $\xi^0 \neq 0$:
\begin{align}
 z^i &= \frac{\xi^i}{\xi^0}, \qquad
 \bar{z}^i = \frac{\bar{\xi}^i}{\bar{\xi}^0}, 
 \qquad (i= 1, 2, \dots, N).
\end{align}
Since 
K\"ahler potential is given by
$
\Phi = \ln (1+|z|^2)
$,
the isometry of $SU(N+1)$ with the
homogeneous coordinates is given by
\begin{align}
 \delta \xi^A &= i \theta^a (T_a)_{AB} \xi^B, 
 \qquad
 \delta \bar{\xi}^A = -i \theta^a \bar{\xi}^B (T_a)_{BA}, 
\end{align} 
where $ \theta^a$ are real parameters,
and its Lie derivative is given by
\begin{align}
& {\cal L}_a = -\left(T_a\right)_{AB}
 \left(
 \xi^B \frac{\partial}{\partial \xi^A}
 -\bar{\xi}^A \frac{\partial}{\partial \bar{\xi}^B}
 \right), \\
 & \left[ {\cal L}_a, {\cal L}_b \right]
 = if_{abc} {\cal L}_c.
\end{align}
Here we introduce the generators $(T_a)_{AB}$ of $su(N+1)$ in the
fundamental representation which satisfy the following relations,
\begin{align}
 [T_a, T_b] &= if_{abc} T_c, \qquad 
 {\rm Tr}~ T_a = 0, \\
 {\rm Tr}~ T_a T_b &= \delta_{ab},  \\
 (T_a)_{AB} (T_a)_{CD} &= \delta_{AD}\delta_{BC}
 -\frac{1}{N+1} \delta_{AB}\delta_{CD},
\end{align}
where $f_{abc}$ is the structure constant of 
$SU(N+1)$, $a=1, 2, \dots, N^2+2N$, and $A, B=0, 1, \dots, N$.  
Generators of the isometry
$SU(N+1)$ in the inhomogeneous coordinates are given as
\begin{align}
 {\cal L}_a = \zeta_a^i \partial_i 
 + \zeta_a^{\bar{i}} \partial_{\bar{i}}
 &= (T_a)_{00} 
 \left(
 z^i \partial_i - \bar{z}^i \partial_{\bar{i}} 
 \right)
 + (T_a)_{0i}
 \left(
 z^iz^j\partial_j + \partial_{\bar{i}}
 \right)  \nonumber \\
	& ~~~~+ (T_a)_{i0}
 \left(
 - \partial_i - \bar{z}^i \bar{z}^j \partial_{\bar{j}}
 \right)
 +(T_a)_{ij}
 \left(
 - z^j \partial_i + \bar{z}^i \partial_{\bar{j}} 
 \right),
\end{align}
and 
\begin{align}
 \zeta^i_a &:= 
 (T_a)_{00} z^i + (T_a)_{0j}z^jz^i - (T_a)_{i0} 
 - (T_a)_{ij}z^j, \\
 \zeta^{\bar{i}}_a &:= 
 -(T_a)_{00} \bar{z}^i + (T_a)_{0i} 
 - (T_a)_{j0}\bar{z}^j\bar{z}^i 
 + (T_a)_{ji}\bar{z}^j.
\end{align}
The quadratic forms of $\zeta^i_a$ and $\zeta^{\bar{i}}_a$ become the metric,
\begin{align}
 \zeta^i_a \zeta^{\bar{j}}_a 
 &= -(1+|z|^2)(\delta_{ij}+z^i \bar{z}^j) = -g^{i\bar{j}}, \\
 \zeta^i_a \zeta^j_a &= 0, \qquad
 \zeta^{\bar{i}}_a \zeta^{\bar{j}}_a = 0.
\end{align}

As we saw in section 2, the Killing vector fields can be represented by
star commutators with the Killing potentials. In the case of 
${\mathbb C}P^N$, using the concrete expressions of the star product in
section 3.1, ${\cal L}_a$ can be written as 
\begin{align}
{\cal L}_a f = -\frac{i}{\hbar} [P_a, f]_* .
\end{align}
$P_a$ are obtained as
\begin{align}
 P_a &= -i(T_a)_{AB}
 \left(
 \frac{\bar{\xi}^A \xi^B}{|\xi|^2} - \delta_{AB}
 \right) \nonumber \\
 &= i(T_a)_{00}\left(
 z^i \partial_i \Phi - 1
 \right)
 - i (T_a)_{0i} \partial_{\bar{i}}\Phi
 -i (T_a)_{i0} \partial_i \Phi
 - i (T_a)_{ij} z^j \partial_i \Phi.
\end{align}
Note that  $P_a$  is determined up to an additional constant.
The Killing potentials 
$P_a$ give a representation of the $su(N+1)$ under the star
commutator, 
\begin{equation}
 [P_a,  P_b]_* = -\hbar f_{abc}  P_c,
\end{equation}
and the bilinear of $P_a$ becomes a constant,
\begin{equation}
 P_a * P_a = -N \left(\frac{1}{N+1} + \hbar \right).
\end{equation}
The Killing potential $P$ in (\ref{P-CP}) can be written in a linear
combination of $P_a$.

The star commutators between $P_a$ and a function $f$
become the Lie derivative ${\cal L}_a f$ of $f$ corresponding to the
generator $T_a$,
\begin{align}
 -\frac{i}{\hbar} [P_a, f]_* &= 
 {\cal L}_a f \nonumber \\
 &= \left[
 (T_a)_{00} 
 \left( z^i \partial_i - \bar{z}^i \partial_{\bar{i}}  \right)
 + (T_a)_{0i}
 \left( z^iz^j\partial_j + \partial_{\bar{i}} \right)  
 \right. 
 \nonumber \\ 
 & \left. ~~~~
 + (T_a)_{i0}
 \left( - \partial_i - \bar{z}^i \bar{z}^j \partial_{\bar{j}} \right)
 +(T_a)_{ij}
 \left( - z^j \partial_i + \bar{z}^i \partial_{\bar{j}} 
 \right) \right] f.
\end{align}
As emphasized before,
since the expression of the star product has the coordinate dependence,
general vector fields do not satisfy the Leibniz rule. However,
the Leibniz rule trivially holds for the Killing vector fields, because 
they are described as the star commutators,
\begin{align}
 {\cal L}_a (f*g) = -\frac{i}{\hbar} [P_a, f*g]_*
 = -\frac{i}{\hbar} [P_a, f]_* *g 
 - \frac{i}{\hbar} f*[P_a, g]
 = ({\cal L}_a f)*g + f*({\cal L}_a g).
\end{align}


\subsection{Differentials on noncommutative $\mathbb{C}H^N$}
As similar to the  $\mathbb{C}P^N$, we give explicit expressions of 
inner derivations given by the Killing potential for $\mathbb{C}H^N$.
The Killing potential 
satisfying $D^i D^j P =0$ and 
$D^{\bar{i}}D^{\bar{j}}P=0$ can be solved as
\begin{align}
 P = \frac{\alpha_i z^i + \bar{\alpha}_i \bar{z}^i 
 + \beta_{ij}\bar{z}^i z^j}{1-|z|^2},
\label{P-CH}
\end{align}
where $\alpha_i$ and $\beta_{ij} = \bar{\beta}_{ji}$ are complex
parameters. 
In the following, we construct inner derivations corresponding to the
isometry transformations.

We first summarize useful facts in the isometry of $\mathbb{C}H^N$.
As homogeneous coordinates of $\mathbb{C}H^N$
we denote 
\begin{align}
 \left\{\zeta^A | A = 0, 1, \cdots, N \right\} &
 = \left\{ \zeta^0, \zeta^i | i= 1, 2, \cdots, N \right\},
\end{align}
and their relation between 
with inhomogeneous coordinates on the chart $\zeta^0 \neq 0$ 
are given by
\begin{align}
 z^i &= \frac{\zeta^i}{\zeta^0}, \qquad
 \bar{z}^i = \frac{\bar{\zeta}^i}{\bar{\zeta}^0}, 
 \qquad (i= 1, 2, \cdots, N).
\end{align}
Since the K\"ahler potential is given by
$\Phi = -\ln (1-|z|^2)$, there is an $SU(1, N)$ isometry.
Let us summarize the notations of $SU(1, N)$.
$SU(1, N)$ transformations preserve 
\begin{equation}
 |\xi|^2 = \eta_{AB} \bar{\xi}^A \xi^B,
\end{equation}
where the metric is
defined by $(\eta_{AB}) = diag.(1, \overbrace{-1, \cdots, -1}^N)$.
In other words, $SU(1, N)$ is defined as
\begin{equation}
 U \in SU(1, N) ~\Longleftrightarrow~  U^\dagger \eta U = \eta, 
  ~~ \det U =1.
\end{equation} 
The Lie algebra $su(1, N)$ is defined by
\begin{equation}
 A \in su(1, N) ~\Longleftrightarrow~ 
  U = e^A \in SU(1, N) ~\Longleftrightarrow~ 
  \eta A^\dagger \eta = -A, ~~ {\rm Tr} A =0.
\end{equation}
As a basis, we choose $(N+1)\times(N+1)$ matrices 
$T_a ~(a=1, 2, \dots, N^2+2N)$ which
satisfy the following relations,
\begin{align}
 & {\rm Tr} T_a = 0, \\
 & \left(T_a^\dagger\right)_{00} = -\left(T_a\right)_{00},~~
 \left(T_a^\dagger\right)_{ij} = -\left(T_a\right)_{ij},
 \nonumber \\
 & \left(T_a^\dagger\right)_{0i} = \left(T_a\right)_{0i},~~
 \left(T_a^\dagger\right)_{i0} = \left(T_a\right)_{i0}, \\
 & {\rm Tr} T_a T_b = h_{ab}, \qquad
 (h_{ab}) = diag.(\overbrace{-1, \cdots, -1}^{N^2}, 
 \overbrace{1, \cdots, 1}^{2N}), \\
 & T_a^\dagger = h_{ab} T_b, \\
 & [T_a, T_b] = f_{abc} T_c, \qquad (f_{abc} \in \mathbb{R}), \\
 & h_{ab} (T_a)_{AB}(T_b)_{CD} =
 \delta_{AD}\delta_{BC} - \frac{1}{N+1} \delta_{AB}\delta_{CD}.	
\end{align}
More explicit form of a basis is given in the appendix \ref{su(1,N)}.
Using these notations, transformations and generators of the isometry
$SU(1, N)$ in the 
homogeneous coordinates are obtained as
\begin{align}
 \delta \xi^A &=  \theta^a (T_a)_{AB} \xi^B, 
 \qquad
 \delta \bar{\xi}^A = \theta^a \bar{\xi}^B (T_a^\dagger)_{BA}, \\ 
 {\cal L}_a &= 
 -\left(T_a\right)_{AB} \xi^B \frac{\partial}{\partial \xi^A}
 -\left(T_a^\dagger\right)_{AB}
 \bar{\xi}^A \frac{\partial}{\partial \bar{\xi}^B}, \\
 \left[ {\cal L}_a, {\cal L}_b \right]
 &= f_{abc} {\cal L}_c.
\end{align}
The generators of the isometry $SU(1, N)$ in the inhomogeneous
coordinates are
\begin{align}
 {\cal L}_a = \zeta_a^i \partial_i + \zeta_a^{\bar{i}} \partial_{\bar{i}}
 &= (T_a)_{00} 
 \left(
 z^i \partial_i - \bar{z}^i \partial_{\bar{i}}
 \right)
 + (T_a)_{0i}
 \left(
 z^iz^j\partial_j -\partial_{\bar{i}} 
 \right)  \nonumber \\
 & ~~~~+ (T_a)_{i0}
 \left(
 - \partial_i + \bar{z}^i \bar{z}^j \partial_{\bar{j}} 
 \right)
 +(T_a)_{ij}
 \left(
 - z^j \partial_i + \bar{z}^i \partial_{\bar{j}} 
 \right),
\end{align}
and
\begin{align}
 \zeta^i_a &:= 
(T_a)_{00} z^i + (T_a)_{0j}z^jz^i - (T_a)_{i0} 
 - (T_a)_{ij}z^j, \\
 \zeta^{\bar{i}}_a &:= 
 -(T_a)_{00} \bar{z}^i - (T_a)_{0i} 
 + (T_a)_{j0}\bar{z}^j\bar{z}^i 
 + (T_a)_{ji}\bar{z}^j.
\end{align}
The quadratic forms of $\zeta^i_a$ and $\zeta^{\bar{i}}_a$ become the metric,
\begin{align}
 \zeta^i_a \zeta^{\bar{j}}_b h_{ab} 
 &= (1-|z|^2)(\delta_{ij} - z^i \bar{z}^j) = g^{i\bar{j}}, \\
 \zeta^i_a \zeta^j_b h_{ab} &= 0, \qquad
 \zeta^{\bar{i}}_a \zeta^{\bar{j}}_b h_{ab} = 0.
\end{align}
As we found in general case 
(or similar to the case of ${\mathbb C}P^N$),
the Killing vector fields are written by commutators
of the Killing potentials,
\begin{align}
{\cal L}_a f = -\frac{i}{\hbar} [P_a, f]_*,
\end{align}
and the $P_a$ are given by
\begin{align}
 P_a =& i(T_a)_{AB}
 \left(
 \frac{\eta_{AC}\bar{\xi}^C \xi^B}{|\xi|^2} 
 - \delta_{AB}
 \right) \nonumber \\
 =& i(T_a)_{00}\left( z^i \partial_i \Phi +1
 \right)
 + i (T_a)_{0i} \partial_{\bar{i}}\Phi
 -i (T_a)_{i0} \partial_i \Phi
 - i (T_a)_{ij} z^j \partial_i \Phi.
\end{align}
Note that  $P_a$  is determined up to an additional constant.
The following formula is also obtained as similar to ${\mathbb C}P^N$:
\begin{align}
P_a * P_b~ h_{ab} =- N \left( \frac{N}{N+1}- \hbar \right).
\end{align}


\subsection{Cyclic property of integration and actions of gauge  theories}

In this section,
we first show the cyclic property of integration, explicitly.
\begin{thm} \label{cycle}
Let  $M$ be $\mathbb{C}P^N$ or $\mathbb{C}H^N$,
and let $F$ and $G$ be arbitrary compact supported bounded smooth functions  on $M$.
Then, the Riemannian volume form is a trace density with respect to the star products
with separation of variables, namely 
we have
\begin{align} \label{CP_CH_cycle}
\int_M F*G \sqrt{g} dz^1 \cdots dz^{N} d\bar{z}^1 \cdots d\bar{z}^N
=\int_M G*F \sqrt{g} dz^1 \cdots dz^{N} d\bar{z}^1 \cdots d\bar{z}^N .
\end{align}
\end{thm}

Note that the star products can be written by using the
Levi-Civita connection $\nabla_i$ and $\nabla_{\bar{i}}$ as
\begin{align}
 F*G &= FG + \sum_{n=1}^\infty c_n(\hbar)
 g^{\bar{i}_1 j_1} \cdots g^{\bar{i}_n j_n}
 \left(\nabla_{\bar{i}_1} \cdots \nabla_{\bar{i}_n} F \right)
 \left(\nabla_{j_1} \cdots \nabla_{j_n} G \right) ,
\end{align}
where
$
 c_n (\hbar) 
 = {\alpha_n (\hbar)} / {n!}
$
for ${\mathbb C}P^N$ and 
 $c_n(\hbar) = \beta_n (\hbar)/n!$ for  $\mathbb{C}H^N$ 
(see \cite{Sako:2012ws}).~
(
Do not confuse the Levi-Civita connection $\nabla_i$ 
with the gauge covariant derivative 
(\ref{covariavt scalar}).)~ 
We use the following relations which hold for the Levi-Civita
connections and the Riemannian curvature tensor on $\mathbb{C}P^N$ and $\mathbb{C}H^N$
(\cite{Kobayashi_Nomizu} p169):
\begin{align}
 & [\nabla_i, \nabla_j] = 0, ~~~~
 [\nabla_{\bar{i}}, \nabla_{\bar{j}}] =0, \\
 & [\nabla_i, \nabla_{\bar{j}}] v_k = {R_{i\bar{j}k}}^l v_l, ~~~~
 [\nabla_i, \nabla_{\bar{j}}] v_{\bar{k}} 
 = {R_{i\bar{j}\bar{k}}}^{\bar{l}} v_{\bar{l}}, \\
 & {R_{i\bar{j}k}}^l = -c(\delta_{kl}g_{i\bar{j}} +\delta_{il}g_{k\bar{j}}), 
 ~~~~
 {R_{\bar{i} j \bar{k}}}^{\bar{l}} = 
 -c(\delta_{kl}g_{j\bar{i}} +\delta_{il}g_{j\bar{k}}), \\
 & \nabla_{m} {R_{i\bar{j}k}}^l = \nabla_{\bar{m}}{R_{i\bar{j}k}}^l
 = \nabla_m {R_{\bar{i} j \bar{k}}}^{\bar{l}}
 = \nabla_{\bar{m}} {R_{\bar{i} j \bar{k}}}^{\bar{l}} = 0.
\end{align}
Here $c=1$ and $c=-1$ are for $\mathbb{C}P^N$ and $\mathbb{C}H^N$,
respectively.
To prove the theorem \ref{cycle},
we use the following lemma.
\begin{lem}
 \label{lemma}
For the arbitrary $C^{\infty} $ function $G$ on $M$,
\begin{align}
 \nabla_{\bar{i}_1} \cdots \nabla_{\bar{i}_n}
 \nabla_{j_1} \cdots \nabla_{j_n} G
 = \nabla_{j_1} \cdots \nabla_{j_n} 
 \nabla_{\bar{i}_1} \cdots \nabla_{\bar{i}_n} G.
 \label{nablaG_lem}
\end{align}
\end{lem}
The proof of this lemma is given in the appendix {\ref{nabla}}.

Theorem \ref{cycle} can be shown easily by using
this lemma. 
 
\begin{proof}
\begin{align}
 \int d\mu F*G &=
 \int d\mu \left[
 FG + \sum_{n=1}^\infty c_n(\hbar)
 g^{\bar{i}_1 j_1} \cdots g^{\bar{i}_n j_n}
 \left(\nabla_{\bar{i}_1} \cdots \nabla_{\bar{i}_n} F \right)
 \left(\nabla_{j_1} \cdots \nabla_{j_n} G \right) \right]
 \nonumber \\
  &= \int d\mu \left[
 GF + \sum_{n=1}^\infty (-1)^n c_n(\hbar)
 g^{\bar{i}_1 j_1} \cdots g^{\bar{i}_n j_n} F
 \left(\nabla_{\bar{i}_1} \cdots \nabla_{\bar{i}_n}
 \nabla_{j_1} \cdots \nabla_{j_n} G \right) \right]
 \nonumber \\
 &= \int d\mu \left[
 GF + \sum_{n=1}^\infty (-1)^n c_n(\hbar)
 g^{\bar{i}_1 j_1} \cdots g^{\bar{i}_n j_n} F
 \left( \nabla_{j_1} \cdots \nabla_{j_n} 
 \nabla_{\bar{i}_1} \cdots \nabla_{\bar{i}_n} G \right) \right]
 \nonumber \\
 &= \int d\mu \left[
 GF + \sum_{n=1}^\infty c_n(\hbar)
 g^{\bar{i}_1 j_1} \cdots g^{\bar{i}_n j_n} 
 \left( \nabla_{\bar{i}_1} \cdots \nabla_{\bar{i}_n} G \right)
 \left(\nabla_{j_1} \cdots \nabla_{j_n} F\right) \right]
 \nonumber \\
 &= \int d\mu G*F
\end{align}
where $d\mu$ is the volume form on $\mathbb{C}P^N$ or $\mathbb{C}H^N$
written in (\ref{CP_CH_cycle}).
\end{proof}

This result is possible to be extended to functions of formal power series
of bounded smooth functions with compact supports.

In section 2, 
we constructed 
a gauge theory on general noncommutative homogeneous K\"ahler manifolds.
In particular, we consider gauge theory on the noncommutative 
${\mathbb C}P^N \approx {\rm SU(N+1)}/{\rm S(U(1)\times U(N))} $
and ${\mathbb C}H^N \approx {\rm SU(1,N)}/{\rm S(U(1)\times U(N))} $
with separation of variables.
In the previous section, the derivations for functions 
on  noncommutative K\"ahler manifolds
with isometry, and concrete expressions of the derivations for 
 ${\mathbb C}P^N$ and ${\mathbb C}H^N$ are constructed.
Using them, a gauge theory with gauge group $G$ on the 
coset space is constructed.
In addition, trace density is given by
usual volume density as we see in this section.
Then the action for the gauge fields is given by
\begin{align}
S_g :=
\int_{{\mathbb C}P^N} \sqrt{g}dz^1 \cdots dz^N
d\bar{z}^1 \cdots d \bar{z}^N
 ~ {\rm tr} \left( {\cal F}_{ab} * {\cal F}_{cd}
\eta^{ac}\eta^{bd} \right) ,
\end{align}
where ${\rm tr}$ is trace for gauge group $G$.
The gauge invariance of the action is guaranteed by
(\ref{Ftrans}) and 
the cyclic symmetry.
The action for the scalar field are same
as (\ref{scalar_action});
\begin{align}
\label{scalar_actionCP}
S_{\phi} = \int_{M} \sqrt{g}dz^1 \cdots dz^N
d\bar{z}^1 \cdots d \bar{z}^N
\{ \nabla_a \phi^{\dagger}* \nabla_b \phi
\eta^{ab} + V(\phi^{\dagger} * \phi )\} .
\end{align}

\section{Conclusions}
We focused on a gauge theory which has derivations given by
order one differential operators, by only considering 
inner derivations which possess vector fields expressions on
general homogeneous K\"ahler manifolds.
As examples, we constructed explicit expressions for these inner derivations 
on ${\mathbb C}P^N$ and ${\mathbb C}H^N$.
For our deformation quantization, we directly proved that
integrations  of $*$-products of functions with
the volume form of the K\"ahler metric of
${\mathbb C}P^N$ and ${\mathbb C}H^N$
have a cyclic property.
We then constructed an action functional having gauge symmetry on these manifolds. 

We note that
the action functionals given in this article are
gauge invariants  not only for noncommutative homogeneous K\"ahler manifolds
but also for the isometry groups of general noncommutative K\"ahler manifolds.
In this sense, gauge theories on general noncommutative K\"ahler manifolds
are constructed in this article.
However, the relation between the usual action of gauge fields
(\ref{taioukankei}) and the normalization (\ref{zeta-g}) 
 is obtained only for noncommutative homogeneous 
K\"ahler manifolds. 
In other words, 
the correspondence between the gauge theories on a commutative space and 
the noncommutative space is clear, and 
it is possible to interpret the noncommutative gauge theory as a
deformation of the commutative gauge theory 
for noncommutative homogeneous K\"ahler manifolds.\\

\noindent
{\bf Acknowledgments} \\
Y.M. was supported in part by JSPS KAKENHI No.23340018
and No.22654011, and A.S. was supported in part by JSPS
KAKENHI No.23540117.
We thank the referee for pointing out the 
relation between section 2.2 and \cite{Muller:2004}.

\appendix
\section{The proof of the lemma \ref{lemma}}
\setcounter{equation}{0}
\label{nabla}

We give the proof of the lemma \ref{lemma},
\begin{align}
 \nabla_{\bar{i}_1} \cdots \nabla_{\bar{i}_n}
 \nabla_{j_1} \cdots \nabla_{j_n} G
 = \nabla_{j_1} \cdots \nabla_{j_n} 
 \nabla_{\bar{i}_1} \cdots \nabla_{\bar{i}_n} G.
 \label{nablaG}
\end{align}

\begin{proof}
When $n=1$, trivially
\begin{align}
 \nabla_{\bar{i}} \nabla_j G = \nabla_j \nabla_{\bar{i}}G.
\end{align}

Assume $ \nabla_{\bar{i}_1} \cdots \nabla_{\bar{i}_{n-1}}
\nabla_{j_1} \cdots \nabla_{j_{n-1}} G
= \nabla_{j_1} \cdots \nabla_{j_{n-1}} 
\nabla_{\bar{i}_1} \cdots \nabla_{\bar{i}_{n-1}} G.$
We here use the following notation for simplicity,
\begin{align}
 [k, l] \equiv g_{\bar{i}_k j_l} 
 \nabla_{\bar{i}_1} \cdots \hat{\nabla}_{\bar{i}_k}
 \cdots \nabla_{\bar{i}_n}
 \nabla_{j_1} \cdots \hat{\nabla}_{j_l} \cdots \nabla_{j_n} G,
\end{align}
where ``$\hat{A}$'' means $A$ is removed.
Then,
\begin{align}
 \nabla_{\bar{i}_1} \cdots \nabla_{\bar{i}_n}
 \nabla_{j_1} \cdots \nabla_{j_n} G
 =& \nabla_{\bar{i}_1} \nabla_{j_1}
 \left(
 \nabla_{\bar{i}_2} \cdots \nabla_{\bar{i}_n}
 \nabla_{j_2} \cdots \nabla_{j_n} G
 \right) \nonumber \\
 & + \sum_{k=2}^n \nabla_{\bar{i}_1} \cdots \nabla_{\bar{i}_{k-1}}
 [\nabla_{\bar{i}_k}, \nabla_{j_1}] 
 \nabla_{\bar{i}_{k+1}} \cdots \nabla_{\bar{i}_n}
 \nabla_{j_2} \cdots \nabla_{j_n} G
 \nonumber \\
 =& \nabla_{\bar{i}_1} \nabla_{j_1}
 \left(
 \nabla_{j_2} \cdots \nabla_{j_n}
 \nabla_{\bar{i}_2} \cdots \nabla_{\bar{i}_n}
 \right)G \nonumber \\
 & + \sum_{k=2}^n \nabla_{\bar{i}_1} \cdots \nabla_{\bar{i}_{k-1}}
 \left[
 \sum_{l=k+1}^n {R_{\bar{i}_k j_1 \bar{i}_l}}^{\bar{p}}
 \nabla_{\bar{i}_{k+1}} \cdots \mathop{\nabla_{\bar{p}}}^{(l)} 
 \cdots \nabla_{\bar{i}_n} \nabla_{j_2} \cdots \nabla_{j_n} G 
 \right. \nonumber \\ 
 & \left. 
 ~~~~ + \sum_{l=2}^n 
 {R_{\bar{i}_k j_1 j_l}}^p
 \nabla_{\bar{i}_{k+1}} \cdots \nabla_{\bar{i}_n}
 \nabla_{j_2} \cdots \mathop{\nabla_p}^{(l)} \cdots \nabla_{j_n} G
 \right]
 \nonumber \\
 =& \nabla_{\bar{i}_1} \nabla_{j_1}
 \left(
 \nabla_{j_2} \cdots \nabla_{j_n}
 \nabla_{\bar{i}_2} \cdots \nabla_{\bar{i}_n}
 \right)G \nonumber \\
 & + c \sum_{k=2}^n 
 \left[
 - \sum_{l=k+1}^n \left([k, 1] + [l, 1]\right)
 + \sum_{l=2}^n \left([k, 1] + [k, l]\right)
 \right] \nonumber \\
  =& \nabla_{\bar{i}_1} \nabla_{j_1}
 \left(
 \nabla_{j_2} \cdots \nabla_{j_n}
 \nabla_{\bar{i}_2} \cdots \nabla_{\bar{i}_n}
 \right) G \nonumber \\
 & + c \sum_{k=2}^n (k-1) [k, 1]
 - c \sum_{k=2}^{n-1} \sum_{l=k+1}^n [l, 1]
 + c \sum_{k=2}^n \sum_{l=2}^n [k, l]
 \nonumber \\
  =& \nabla_{\bar{i}_1} \nabla_{j_1}
 \left(
 \nabla_{j_2} \cdots \nabla_{j_n}
 \nabla_{\bar{i}_2} \cdots \nabla_{\bar{i}_n}
 \right)G \nonumber \\
 & + c \sum_{k=2}^n (k-1) [k, 1]
 - c \sum_{l=3}^n \sum_{k=2}^{l-1} [l, 1]
 + c \sum_{k=2}^n \sum_{l=2}^n [k, l]
 \nonumber \\
  =& \nabla_{\bar{i}_1} \nabla_{j_1}
 \left(
 \nabla_{j_2} \cdots \nabla_{j_n}
 \nabla_{\bar{i}_2} \cdots \nabla_{\bar{i}_n}
 \right)G \nonumber \\
 & + c \sum_{k=2}^n (k-1) [k, 1]
 - c \sum_{l=3}^n (l-2) [l, 1]
 + c \sum_{k=2}^n \sum_{l=2}^n [k, l]
 \nonumber \\
   =& \nabla_{\bar{i}_1} \nabla_{j_1}
 \left(
 \nabla_{j_2} \cdots \nabla_{j_n}
 \nabla_{\bar{i}_2} \cdots \nabla_{\bar{i}_n}
 \right)G 
 + c \sum_{k=2}^n [k, 1]
 + c \sum_{k=2}^n \sum_{l=2}^n [k, l].
\end{align}
Next, the first term in the last expression, 
$\nabla_{\bar{i}_1} \nabla_{j_1} \left( \nabla_{j_2} \cdots \nabla_{j_n}
 \nabla_{\bar{i}_2} \cdots \nabla_{\bar{i}_n} \right)G$
becomes
\begin{align}
 \nabla_{\bar{i}_1} \nabla_{j_1} \nabla_{j_2} \cdots \nabla_{j_n}
 \nabla_{\bar{i}_2} \cdots \nabla_{\bar{i}_n} G
 =& \nabla_{j_1} \cdots \nabla_{j_n}
 \nabla_{\bar{i}_1} \cdots \nabla_{\bar{i}_n}G
 \nonumber \\
 & + \sum_{k=1}^n \nabla_{j_1} \cdots \nabla_{j_{k-1}}
 \left[\nabla_{\bar{i}_1}, \nabla_{j_k}\right]
 \nabla_{j_{k+1}} \cdots \nabla_{j_n}
 \nabla_{\bar{i}_2} \cdots \nabla_{\bar{i}_n} G
 \nonumber \\
 =& \nabla_{j_1} \cdots \nabla_{j_n}
 \nabla_{\bar{i}_1} \cdots \nabla_{\bar{i}_n}G
 \nonumber \\
 & + \sum_{k=1}^n \nabla_{j_1} \cdots \nabla_{j_{k-1}}
 \left[
 \sum_{l=k+1}^n {R_{\bar{i}_1 j_k j_l}}^p
 \nabla_{j_{k+1}} \cdots \mathop{\nabla_p}^{(l)} \cdots \nabla_{j_n}
 \nabla_{\bar{i}_2} \cdots \nabla_{\bar{i}_n} G 
 \right. \nonumber \\
 & \left. ~~~~~
 + \sum_{l=2}^n {R_{\bar{i}_1 j_k \bar{i}_l}}^{\bar{p}}
 \nabla_{j_{k+1}} \cdots \nabla_{j_n}
 \nabla_{\bar{i}_2} \cdots \mathop{\nabla_{\bar{p}}}^{(l)} 
 \cdots \nabla_{\bar{i}_n} G
 \right]
 \nonumber \\
 =& \nabla_{j_1} \cdots \nabla_{j_n}
 \nabla_{\bar{i}_1} \cdots \nabla_{\bar{i}_n}G
 \nonumber \\
 & + c \sum_{k=1}^n
 \left[
 \sum_{l=k+1}^n \left([1, k] + [1, l]\right)
 - \sum_{l=2}^n \left([1, k] + [l, k]\right)
 \right]
 \nonumber \\
 =& \nabla_{j_1} \cdots \nabla_{j_n}
 \nabla_{\bar{i}_1} \cdots \nabla_{\bar{i}_n}G
 \nonumber \\
 & - c \sum_{k=2}^n (k-1) [1, k]
 + c \sum_{k=1}^n \sum_{l=k+1}^n  [1, l]
 - c \sum_{k=1}^n \sum_{l=2}^n [l, k]
 \nonumber \\
 =& \nabla_{j_1} \cdots \nabla_{j_n}
 \nabla_{\bar{i}_1} \cdots \nabla_{\bar{i}_n}G
 - c \sum_{l=2}^n [l, 1]
 - c \sum_{k=2}^n \sum_{l=2}^n [l, k].
\end{align}
This completes the proof for the lemma.

\end{proof}

\section{A basis of $su(1, N)$}
\label{su(1,N)}
\setcounter{equation}{0}

A concrete basis of $su(1, N)$, $T_a, (a = 1, 2, \cdots, (N+1)^2 -1)$ 
is given as follows;
\begin{align}
 \{T_a\} = \{I_{ij},~ J_{ij},~ H_k, 
 I_{0i},~ J_{0i}\},
\end{align}
where $i, j, k = 1, 2, \cdots, N$ 
and $i<j$ in $I_{ij}, J_{ij}$.
\begin{align}
 I_{ij} &= \frac{1}{\sqrt{2}} (E_{ij} - E_{ji}), \\
 J_{ij} &= \frac{i}{\sqrt{2}} (E_{ij} + E_{ji}), \\
 H_k &= \frac{i}{\sqrt{k(k+1)}}
 \left(
 \sum_{i=1}^k E_{ii} - k E_{k+1, k+1}
 \right), ~~~~ (E_{N+1, N+1} = E_{00}), \\
 I_{0i} &= \frac{1}{\sqrt{2}}(E_{i0} + E_{0i}), \\
 J_{0i} &= \frac{i}{\sqrt{2}}(E_{i0} - E_{0i}),
\end{align}
where $(E_{AB})_{CD} = \delta_{AC} \delta_{BD}$ and 
$A, B, C, D = 0, 1, \dots N$. 
$I_{ij}, J_{ij}, H_k$ are anti-hermitian and $I_{0i}, J_{0i}$ are hermitian.



\end{document}